\documentclass[aps,preprint]{revtex4}%
\usepackage{amsfonts}
\usepackage{amsmath}
\usepackage{amssymb}
\usepackage{graphicx}%
\providecommand{\U}[1]{\protect\rule{.1in}{.1in}}

\newtheorem{theorem}{Theorem}

\newtheorem{definition}[theorem]{Definition}

\newenvironment{proof}[1][Proof]{\noindent\textbf{#1.} }{\ \rule{0.5em}{0.5em}}
\begin{document}
\title{S-Expansion of Higher-Order Lie Algebras}
\author{Ricardo Caroca$^{1,2}$}
\email{rcaroca@ucsc.cl}
\author{Nelson Merino$^{1}$}
\email{nemerino@udec.cl}
\author{Patricio Salgado$^{1}$}
\email{pasalgad@udec.cl}
\date{\today}
\begin{abstract}
By means of a generalization of the S-expansion method we construct a
procedure to obtain expanded higher-order Lie algebras. It is shown that the
direct product between an Abelian semigroup S and a higher-order Lie algebra
($\mathcal{G},\left[  ,...,\right]  )$ is also a higher-order Lie algebra.
\ From this S-expanded Lie algebra are obtained resonant submultialgebras and
reduced multialgebras of a resonant submultialgebra.

\end{abstract}
\affiliation{$^{1}$Departamento de F\'{\i}sica, Universidad de Concepci\'{o}n, Casilla 160-C, Concepci\'{o}n, Chile\\
$^{2}$Departamento de Matematica y F\'{\i}sica Aplicadas, Universidad Cat\'{o}lica
de la Santisima Concepci\'{o}n, Alonso de Rivera 2850, Concepci\'on, Chile }
\maketitle
\tableofcontents

\section{Introduction}

Higher-order (or multibracket) simple Lie algebras \cite{deazcarraga},
\cite{deazcarraga01}, \cite{deazcarraga0} are generalized ordinary Lie
algebras. Their structure constants are given by Lie algebra cohomology
\ cocycles which, by virtue of being such, satisfy a suitable generalization
of the Jacobi identity. \bigskip

As is noted in ref \cite{deazcarraga}, \cite{deazcarraga0} it could be
interesting to find applications of these higher-order Lie algebras to know
whether the cohomological restrictions which determine and condition their
existence have a physical significance. \ Lie algebra cohomology arguments
have already been very useful in various physical problems as in the
description of anomalies or in the construction of the Wess-Zumino terms
required in the action of extended supersymmetric objects. \ Other questions
may be posed from a purely mathematical point of view. From the discussion in
Sect.4 of ref. \cite{deazcarraga} we know that a representation of a simple
Lie algebra may not be a representation for the associated higher-order Lie
algebras. Thus, the representation theory of higher-order algebras requires a
separate analysis. A very interesting open problem from a structural point of
view is the expansions of higher-order Lie algebras, which will take us
outside the domain of the simple ones.

The purpose of this paper is to show that the S-expansion method developed in
ref. \cite{sexpansion} (see also \cite{deazcarraga1},\cite{hatsuda},
\cite{deazcarraga2}) can be generalized so that it permits obtaining expanded
higher-order Lie algebras.

The paper is organized as follows: In section 2 we shall review some aspects
of higher-order Lie algebras. The main point of this section is to display the
differences between ordinary Lie algebras and higher-order Lie algebras and to
generalize the definitions of higher-order Lie subalgebras and higher-order
reduced Lie algebras. In section 3 we generalize the S-expansion method and we
show that it is possible to obtain higher-order expanded Lie algebras. In
section 4 \ is shown that, under determined conditions, relevant higher-order
Lie subalgebras can be extracted from the S-expanded higher-order Lie algebras.

\section{ Higher-order Lie algebras}

In this section we shall review some aspects of higher-order Lie algebras. The
main point of this section is to display the differences between ordinary Lie
algebras and higher-order Lie algebras and to generalize the concepts of
subalgebra and reduced Lie algebra of ref. \cite{sexpansion}.

\begin{definition}
An algebra is defined as a pair $\left(  G,\bullet\right)  $ where $G$ is a
finite dimensional vector space,\ and $\ \bullet:G\times G\rightarrow G$ is a
rule of composition defined over the vector space.
\end{definition}

\begin{definition}
A Lie algebra $\mathcal{G}$ is defined by the pair $\left(  G,\left[
,\right]  \right)  $ where $G$ is a finite dimensional vector space , with
basis $\left\{  T_{A}\right\}  _{A=1}^{\dim G}$, over the field $K$ of real or
complex numbers; and $\left[  ,\right]  $ is a rule of composition $\left(
T_{A_{1}},T_{A_{2}}\right)  \rightarrow\left[  T_{A_{1}},T_{A_{2}}\right]  \in
G$ which satisfies the following axioms:
\end{definition}

\begin{itemize}
\item $\left[  \alpha T_{A_{1}}+\beta T_{A_{2}},T_{A_{3}}\right]
=\alpha\left[  T_{A_{1}},T_{A_{3}}\right]  +\beta\left[  T_{A_{2}},T_{A_{3}%
}\right]  $ \textit{for} $\alpha,\beta\in K$ \ \textit{\ (linearity),}

\item $\left[  T_{A_{1}},T_{A_{2}}\right]  =-\left[  T_{A_{2}},T_{A_{1}%
}\right]  $ \ \ \ \ \ \ \ \ $\mathit{\forall}$ $T_{A_{1}},T_{A_{2}}\in G$
\textit{\ \ (antisymmetry),}

\item $\left[  \left[  T_{A_{1}},T_{A_{2}}\right]  ,T_{A_{3}}\right]  +\left[
\left[  T_{A_{2}},T_{A_{3}}\right]  ,T_{A_{1}}\right]  +\left[  \left[
T_{A_{3}},T_{A_{1}}\right]  ,T_{A_{2}}\right]  =0$,

for all $T_{A_{1}},T_{A_{2}},T_{A_{3}}\in G$\ \textit{(Jacobi identity).}
\end{itemize}

The Jacobi identity (JI) can be re-written
\begin{equation}
\frac{1}{1!}\frac{1}{2!}%
{\displaystyle\sum\limits_{\sigma\in S_{3}}}
\left(  -1\right)  ^{\pi\left(  \sigma\right)  }\left[  \left[  T_{A_{\sigma
\left(  1\right)  }},T_{A_{\sigma\left(  2\right)  }}\right]  ,T_{A_{\sigma
\left(  3\right)  }}\right]  =0. \label{i1}%
\end{equation}
where $S_{3}$ is the permutation group of three elements and $\pi\left(
\sigma\right)  $ is the parity of the permutation $\sigma$.

\begin{definition}
Let $\mathcal{G}$\ be a Lie algebra. A n-bracket $\left[  ,...,\right]  $ or
skew-symmetric Lie multibracket is a Lie algebra valued n-linear
skew-symmetric mapping $\left[  ,...,\right]  :\mathcal{G}\overset{n}%
{\times...\times}\mathcal{G}\rightarrow\mathcal{G}$,
\begin{equation}
\left(  T_{A_{1}},...,T_{A_{n}}\right)  \rightarrow\left[  T_{A_{1}%
},...,T_{A_{n}}\right]  =C_{A_{1}...A_{n}}^{B}T_{B} \label{m1}%
\end{equation}
where the constants $C_{A_{1}...A_{n}}^{B}$ are called higher-order structure
constants which are completely antisymmetric in the indices $A_{1}...A_{n}$.
\end{definition}

To define higher-order Lie algebras we need to find the generalization of the
Jacobi identity. We postulate that the generalization of the left hand side of
eq. (\ref{i1}) is given by%
\begin{equation}
\frac{1}{\left(  n-1\right)  !}\frac{1}{n!}%
{\displaystyle\sum\limits_{\sigma\in S_{2n-1}}}
\left(  -1\right)  ^{\pi\left(  \sigma\right)  }\left[  \left[  T_{A_{\sigma
\left(  1\right)  }},...,T_{A_{\sigma\left(  n\right)  }}\right]
,T_{A_{\sigma\left(  n+1\right)  }},...,T_{A_{\sigma\left(  2n-1\right)  }%
}\right]  \label{mm0}%
\end{equation}
However we must find the conditions under which is possible the vanishing of
the right hand side. Let $T_{A}$ be the basis of the algebra in a
representation of $\mathcal{G}$. \ Then is possible to realize the
multibracket as%

\begin{align}
\left[  T_{A_{1}},...,T_{A_{n}}\right]   &  =\varepsilon_{A_{1}...A_{n}%
}^{B_{1}...B_{n}}T_{B_{1}}...T_{B_{n}}\label{m2}\\
&  =%
{\displaystyle\sum\limits_{\sigma\in S_{n}}}
\left(  -1\right)  ^{\pi\left(  \sigma\right)  }T_{A_{\sigma\left(  1\right)
}}...T_{A_{\sigma\left(  n\right)  }}\text{,}\nonumber
\end{align}
where $S_{n}$ is the permutation group of $n$ element and $\pi\left(
\sigma\right)  $ is the parity of the permutation $\sigma$. \ In the appendix
we will show that the realization (\ref{m2})\ of the multibracket satisfy the identity%

\begin{align}
&  \frac{1}{\left(  n-1\right)  !}\frac{1}{n!}%
{\displaystyle\sum\limits_{\sigma\in S_{2n-1}}}
\left(  -1\right)  ^{\pi\left(  \sigma\right)  }\left[  \left[  T_{A_{\sigma
\left(  1\right)  }},...,T_{A_{\sigma\left(  n\right)  }}\right]
,T_{A_{\sigma\left(  n+1\right)  }},...,T_{A_{\sigma\left(  2n-1\right)  }%
}\right] \label{m3}\\
&  =\left\{
\begin{array}
[c]{c}%
0\text{ \ \ \ \ \ \ \ \ \ \ \ \ \ \ \ \ \ \ \ \ \ \ \ \ \ \ \ ,\ }n\text{
even}\\
n\left[  T_{A_{1}},...,T_{A_{2n-1}}\right]  \text{, \ \ \ }n\text{ odd.}%
\end{array}
\right\}  .\nonumber
\end{align}
This means that is possible to obtain a generalization of the Jacobi identity
for $n$ even. For $n$ odd we obtain an identity which contains a combination
of multibrackets of different orders. Thus we can postulate that
\cite{deazcarraga}
\begin{equation}
\frac{1}{\left(  n-1\right)  !}\frac{1}{n!}%
{\displaystyle\sum\limits_{\sigma\in S_{2n-1}}}
\left(  -1\right)  ^{\pi\left(  \sigma\right)  }\left[  \left[  T_{A_{\sigma
\left(  1\right)  }},...,T_{A_{\sigma\left(  n\right)  }}\right]
,T_{A_{\sigma\left(  n+1\right)  }},...,T_{A_{\sigma\left(  2n-1\right)  }%
}\right]  =0\text{,} \label{m1_2}%
\end{equation}
is the appropriate generalization of the Jacobi Identity for $n$ even. This
identity implies the following condition on the structure constants
$C_{A_{1}...A_{n}}^{B}$:
\begin{equation}
\varepsilon_{A_{1}...A_{2n-1}}^{B_{1}...B_{2n-1}}C_{B_{1}...B_{n}}%
^{C}C_{CB_{n+1}...B_{2n-1}}^{D}=0 \label{m4}%
\end{equation}
which is the generalization of the Jacobi condition \cite{deazcarraga}.

By analogy with the standard Lie algebra, we may now give the following
definition \cite{deazcarraga}:

\begin{definition}
Let $\mathcal{G}$ be a Lie algebra and let $n$ be even. A higher-order Lie
algebra or multialgebra on $\mathcal{G}$ is the algebra defined by the pair
$\left(  \mathcal{G},\left[  ,...,\right]  \right)  $ where the multibracket
$\left[  ,...,\right]  $ (\ref{m1}) is multilinear, antisymmetric and
satisfies the generalized Jacobi identity (\ref{m1_2}); and where the
higher-order structure constants satisfy the generalized Jacobi condition
(\ref{m4}).
\end{definition}

The following definition generalizes the concept of Subalgebra:

\begin{definition}
(\textbf{Submultialgebra):} Let $\left(  \mathcal{G},\left[  ,...,\right]
\right)  $ be a multialgebra, and consider the Lie algebra $\mathcal{G}$ of
the form $\mathcal{G}=V_{0}\oplus V_{1}$. The subspace $\left(  V_{0},\left[
,...,\right]  \right)  $ will be called a submultialgebra of $\left(
\mathcal{G},\left[  ,...,\right]  \right)  $ if it satisfies
\begin{equation}
\left[  V_{0},V_{0},...,V_{0}\right]  \subset V_{0}. \label{sm1}%
\end{equation}

\end{definition}

The existence of submultialgebras is reflected in certain definite
restrictions on the structure constants. Let $C_{A_{1}...A_{n}}^{B}$ be the
generalized structure constants of the multialgebra $\left(  \mathcal{G}%
,\left[  ,...,\right]  \right)  .$ If $\left\{  T_{A_{i}}\right\}  $,
$\left\{  T_{a_{i}^{0}}\right\}  $ and $\left\{  T_{a_{i}^{1}}\right\}  $
denote the bases of $\mathcal{G}$, $V_{0}$ and $V_{1}$ respectively, where
$A_{i}=1,...,\dim\mathcal{G}$, $a_{i}^{0}=1,...,\dim V_{0}$ and $a_{i}%
^{1}=\dim V_{0}+1,...,\dim\mathcal{G}$, then the condition (\ref{sm1}) can be
expressed as%
\begin{equation}
C_{a_{1}^{0}...a_{n}^{0}}^{b^{1}}=0 \label{sm2}%
\end{equation}
for $a_{1}^{0}...a_{n}^{0}\leq\dim V_{0}$ and $b^{1}\geq\dim V_{0}+1.$ In
fact, If $V_{0}$ is a submultialgebra then $\left[  V_{0},V_{0},...,V_{0}%
\right]  \subset V_{0}.$ This mean that%
\begin{equation}
\left[  T_{a_{1}^{0}},...,T_{a_{n}^{0}}\right]  =C_{a_{1}^{0}...a_{n}^{0}%
}^{b^{0}}T_{b^{0}}. \label{sm3}%
\end{equation}
i.e. for $\dim V_{0}$ $<$ $b^{1}<\dim G$ we have $C_{a_{1}^{0}...a_{n}^{0}%
}^{b^{1}}=0.$

The following theorem generalizes the concept of reduction of Lie algebras of
ref. \cite{sexpansion} to higher-order Lie algebras.

\begin{theorem}
(Reduced Multialgebra): Let $\left(  \mathcal{G},\left[  ,...,\right]
\right)  $ be a multialgebra, and consider the Lie algebra $G$ of the form
$G=V_{0}\oplus V_{1},$ with $\left\{  T_{A_{i}}\right\}  $ being a basis for
$G$, $\left\{  T_{a_{i}^{0}}\right\}  $ a basis for $V_{0}$ and $\left\{
T_{a_{i}^{1}}\right\}  $ a basis for $V_{1}$. If the condition
\begin{equation}
\left[  V_{1},V_{0},...,V_{0}\right]  \subset V_{1}, \label{mr1}%
\end{equation}
is satisfied, then the structure constants $C_{e^{1}B_{n+1}...B_{2n-1}}%
^{d^{0}}$ are cero, which lead to that the structure constants $C_{a_{1}%
^{0}...a_{n}^{0}}^{b^{0}}$ satisfy the generalized Jacobi condition by
themselves, and therefore
\begin{equation}
\left[  T_{a_{1}^{0}},...,T_{a_{n}^{0}}\right]  =C_{a_{1}^{0}...a_{n}^{0}%
}^{b^{0}}T_{b^{0}} \label{mr2}%
\end{equation}
corresponds by itself to a high-order Lie algebra. This algebra, with
structure constants $C_{a_{1}^{0}...a_{n}^{0}}^{\text{ \ \ \ \ \ \ \ \ \ \ }%
b^{0}}$, is called a reduced multialgebra of $\left(  \mathcal{G},\left[
,...,\right]  \right)  $ and is symbolized as $\left\vert V_{0},\left[
,...,\right]  \right\vert $.
\end{theorem}

\begin{proof}
If the condition
\[
\left[  V_{1},V_{0},...,V_{0}\right]  \subset V_{1}%
\]
is satisfied, we have
\[
\left[  T_{a_{1}^{0}},...,T_{a_{n}^{0}}\right]  =C_{a_{1}^{0}...a_{n}^{0}%
}^{b^{0}}T_{b^{0}}+C_{a_{1}^{0}...a_{n}^{0}}^{b^{1}}T_{b^{1}}%
\]%
\[
\left[  T_{b^{1}a_{1}^{0}},...,T_{a_{n-1}^{0}}\right]  =C_{b^{1}a_{1}%
^{0}...a_{n-1}^{0}}^{c^{1}}T_{c^{1}}%
\]%
\begin{equation}
\left[  T_{b_{1}^{1}},...,T_{b_{n}^{1}}\right]  =C_{b_{1}^{1}...b_{n}^{1}%
}^{c^{0}}T_{c^{0}}+C_{b_{1}^{1}...b_{n}^{1}}^{c^{1}}T_{c^{1}} \label{mr2'}%
\end{equation}
The structure constant of $G$ satisfy the Jacobi identity\textbf{ }%
\begin{equation}
\varepsilon_{A_{1}...A_{2n-1}}^{B_{1}...B_{2n-1}}C_{B_{1}...B_{n}}%
^{C}C_{CB_{n+1}...B_{2n-1}}^{D}=0\text{.} \label{mr3}%
\end{equation}
If $\mathcal{G}=V_{0}\oplus V_{1}$ y $\left\{  T_{A_{i}}\right\}  $, $\left\{
T_{a_{i}^{0}}\right\}  $, y $\left\{  T_{a_{i}^{1}}\right\}  $ are the
corresponding bases of $\mathcal{G}$, $V_{0}$, y $V_{1}$ (where $A_{i}%
=1,...,\dim\mathcal{G}$, $\ a_{i}^{0}=1,...,\dim V_{0}$ \ and $a_{i}^{1}=\dim
V_{0}+1,...,\dim\mathcal{G}$), then the generalized Jacobi condition on
$V_{0}$ is given by%
\begin{equation}
\varepsilon_{a_{1}^{0}...a_{2n-1}^{0}}^{B_{1}...B_{2n-1}}C_{B_{1}...B_{n}}%
^{E}C_{EB_{n+1}...B_{2n-1}}^{d^{0}}=0 \label{mr4}%
\end{equation}
which can be re-written as
\begin{equation}
\varepsilon_{a_{1}^{0}...a_{2n-1}^{0}}^{B_{1}...B_{2n-1}}C_{B_{1}...B_{n}%
}^{e^{0}}C_{e^{0}B_{n+1}...B_{2n-1}}^{d^{0}}+\varepsilon_{a_{1}^{0}%
...a_{2n-1}^{0}}^{B_{1}...B_{2n-1}}C_{B_{1}...B_{n}}^{e^{1}}C_{e^{1}%
B_{n+1}...B_{2n-1}}^{d^{0}}=0\text{.} \label{mr5}%
\end{equation}
We consider now the indices $B_{1}...B_{2n-1}$. \ If one of these indices
takes on a value in $V_{1},$ we have
\begin{equation}
\varepsilon_{a_{1}^{0}.......a_{2n-1}^{0}}^{b_{1}^{1}b_{2}^{0}...b_{2n-1}^{0}%
}=\left\vert
\begin{array}
[c]{cccc}%
\delta_{a_{1}^{0}}^{b_{1}^{1}} & \delta_{a_{1}^{0}}^{b_{2}^{0}} & \ldots &
\delta_{a_{1}^{0}}^{b_{2n-1}^{0}}\\
\delta_{a_{2}^{0}}^{b_{1}^{1}} & \delta_{a_{2}^{0}}^{b_{2}^{0}} & \ldots &
\delta_{a_{2}^{0}}^{b_{2n-1}^{0}}\\
\vdots & \vdots & \ddots & \vdots\\
\delta_{a_{2n-1}^{0}}^{b_{1}^{1}} & \delta_{a_{2n-1}^{0}}^{b_{2}^{0}} & \ldots
& \delta_{a_{2n-1}^{0}}^{b_{2n-1}^{0}}%
\end{array}
\right\vert =\left\vert
\begin{array}
[c]{cccc}%
0 & \delta_{a_{1}^{0}}^{b_{2}^{0}} & \ldots & \delta_{a_{1}^{0}}^{b_{2n-1}%
^{0}}\\
0 & \delta_{a_{2}^{0}}^{b_{2}^{0}} & \ldots & \delta_{a_{2}^{0}}^{b_{2n-1}%
^{0}}\\
\vdots & \vdots & \ddots & \vdots\\
0 & \delta_{a_{2n-1}^{0}}^{b_{2}^{0}} & \ldots & \delta_{a_{2n-1}^{0}%
}^{b_{2n-1}^{0}}%
\end{array}
\right\vert =0. \label{mr6}%
\end{equation}
From (\ref{mr6}) we can see that a column of the determinant is zero and
therefore $\varepsilon_{a_{1}^{0}.......a_{2n-1}^{0}}^{b_{1}^{1}b_{2}%
^{0}...b_{2n-1}^{0}}=0$. \ Similarly, any permutation on the set $\left(
b_{1}^{1}b_{2}^{0}...b_{2n-1}^{0}\right)  $ in $\varepsilon_{a_{1}%
^{0}.......a_{2n-1}^{0}}^{b_{1}^{1}b_{2}^{0}...b_{2n-1}^{0}}$ will be null. If
two indices of the set $\left(  B_{1}...B_{2n-1}\right)  $ take on values in
$V_{1},$ we have
\begin{equation}
\varepsilon_{a_{1}^{0}...........a_{2n-1}^{0}}^{b_{1}^{1}b_{2}^{1}b_{3}%
^{0}...b_{2n-1}^{0}}=\left\vert
\begin{array}
[c]{ccccc}%
\delta_{a_{1}^{0}}^{b_{1}^{1}} & \delta_{a_{1}^{0}}^{b_{2}^{1}} &
\delta_{a_{1}^{0}}^{b_{3}^{0}} & \ldots & \delta_{a_{1}^{0}}^{b_{2n-1}^{0}}\\
\delta_{a_{2}^{0}}^{b_{1}^{1}} & \delta_{a_{2}^{0}}^{b_{2}^{1}} &
\delta_{a_{2}^{0}}^{b_{3}^{0}} & \ldots & \delta_{a_{2}^{0}}^{b_{2n-1}^{0}}\\
\vdots & \vdots & \vdots & \ddots & \vdots\\
\delta_{a_{2n-1}^{0}}^{b_{1}^{1}} & \delta_{a_{2n-1}^{0}}^{b_{2}^{1}} &
\delta_{a_{2n-1}^{0}}^{b_{3}^{0}} & \ldots & \delta_{a_{2n-1}^{0}}%
^{b_{2n-1}^{0}}%
\end{array}
\right\vert =\left\vert
\begin{array}
[c]{ccccc}%
0 & 0 & \delta_{a_{1}^{0}}^{b_{3}^{0}} & \ldots & \delta_{a_{1}^{0}}%
^{b_{2n-1}^{0}}\\
0 & 0 & \delta_{a_{2}^{0}}^{b_{3}^{0}} & \ldots & \delta_{a_{2}^{0}}%
^{b_{2n-1}^{0}}\\
\vdots & \vdots & \vdots & \ddots & \vdots\\
0 & 0 & \delta_{a_{2n-1}^{0}}^{b_{3}^{0}} & \ldots & \delta_{a_{2n-1}^{0}%
}^{b_{2n-1}^{0}}%
\end{array}
\right\vert =0. \label{mr7}%
\end{equation}
From (\ref{mr7}) we can see that a column of the determinant is zero and
therefore $\varepsilon_{a_{1}^{0}...........a_{2n-1}^{0}}^{b_{1}^{1}b_{2}%
^{1}b_{3}^{0}...b_{2n-1}^{0}}=0$. In general the number of null columns
increase with the number of indices of \ the set $\left(  B_{1}...B_{2n-1}%
\right)  $, which take on values in $V_{1}$. Thus, the equation (\ref{mr5}) is
then given by%
\begin{equation}
\varepsilon_{a_{1}^{0}...a_{2n-1}^{0}}^{b_{1}^{0}...b_{2n-1}^{0}}C_{b_{1}%
^{0}...b_{n}^{0}}^{e^{0}}C_{e^{0}B_{n+1}...B_{2n-1}}^{d^{0}}+\varepsilon
_{a_{1}^{0}...a_{2n-1}^{0}}^{b_{1}^{0}...b_{2n-1}^{0}}C_{b_{1}^{0}...b_{n}%
^{0}}^{e^{1}}C_{e^{1}B_{n+1}...B_{2n-1}}^{d^{0}}=0\text{.} \label{mr8}%
\end{equation}
From (\ref{mr8}) we can see that the structure constant $C_{a_{1}^{0}%
...a_{n}^{0}}^{b^{0}}$ satisfy the generalized Jacobi identity by themselves
in two cases:
\end{proof}

\begin{itemize}
\item When $C_{b_{1}^{0}...b_{n}^{0}}^{e^{1}}=0,$ i.e., when $V_{0}$ is a submultialgebra

\item When $C_{e^{1}B_{n+1}...B_{2n-1}}^{d^{0}}=0,$ i.e., when $\left[
V_{1},V_{0},...,V_{0}\right]  \subset V_{1}.$ This means that in this case the
structure constant $C_{a_{1}^{0}...a_{n}^{0}}^{b^{0}}$ satisfy the generalized
Jacobi identity and
\begin{equation}
\left[  T_{a_{1}^{0}},...,T_{a_{n}^{0}}\right]  =C_{a_{1}^{0}...a_{n}^{0}%
}^{b^{0}}T_{b^{0}} \label{mr9}%
\end{equation}
correspond by itself to a higher order Lie algebra. It is interesting to note
that a reduced multialgebra $\left\vert V_{0},\left[  ,...,\right]
\right\vert $ does not correspond to a submultialgebra of $\left(
\mathcal{G},\left[  ,...,\right]  \right)  $.
\end{itemize}

\begin{definition}
The Lie multialgebra obtained from the condition $\left[  V_{1},V_{0}%
,...,V_{0}\right]  \subset V_{1}$ i.e., with $C_{e^{1}B_{n+1}...B_{2n-1}%
}^{d^{0}}=0$ is called a reduced multialgebra of \ $G$ and will be symbolized
as $\mid V_{0}\mid.$
\end{definition}

\section{$S$-expansion of Higher-Order Lie Algebras}

In this section we shall review some aspects of the S-expansion procedure
introduced in ref. \cite{sexpansion}. The main point of this section and of
this paper is to show that the generalization of the S-expansion method
permits obtaining S-expanded higher-order Lie algebras. \ 

\subsection{$S$-Expansion of Lie Algebras}

The $S$-expansion method is based on combining the structure constants of
\ the Lie algebra\ $\left(  \mathcal{G},\left[  ,\right]  \right)  $ with the
inner law of a semigroup $S$ to define the Lie bracket of a new, $S$-expanded
algebra. Let $S=\left\{  \lambda_{\alpha}\right\}  $ be a finite Abelian
semigroup endowed with a commutative and associative composition law $S\times
S\rightarrow S,$ $\left(  \lambda_{\alpha},\lambda_{\beta}\right)
\mapsto\lambda_{\alpha}\lambda_{\beta}=K_{\alpha\beta}^{\text{ \ \ \ \ }%
\gamma}\lambda_{\gamma}.$ \textbf{Let }the pair $\left(  \mathcal{G},\left[
,\right]  \right)  $ a Lie algebra where $G$ is a finite dimensional vector
space, with basis $\left\{  T_{A}\right\}  _{A=1}^{\dim\mathcal{G}}$, over the
field $K$; and $\left[  ,\right]  $ is a ruler of compostion $G\times
G\longrightarrow G,$ $\left(  T_{A_{i}},T_{A_{j}}\right)  \longrightarrow
\left[  T_{A_{i}},T_{A_{j}}\right]  =C_{A_{i}A_{j}}^{A_{k}}T_{A_{k}}.$ The
direct product $G=S\otimes G$ is defined as the Cartesian product set
\begin{equation}
\mathfrak{G}=S\times\mathcal{G}=\left\{  T_{\left(  A,\alpha\right)  }%
=\lambda_{\alpha}T_{A}\text{ : }\lambda_{\alpha}\in S\text{ , }T_{A}%
\in\mathcal{G}\right\}  \label{s1}%
\end{equation}
endowed with a composition law $\left[  ,\right]  _{S}$ $:G\times G\rightarrow
G$ defined by%
\begin{equation}
\left[  T_{\left(  A,\alpha\right)  },T_{\left(  B,\beta\right)  }\right]
_{S}=:\lambda_{\alpha}\lambda_{\beta}\left[  T_{A},T_{B}\right]
=K_{\alpha\beta}^{\gamma}C_{AB}^{C}\lambda_{\gamma}T_{C}=C_{\left(
A,\alpha\right)  \left(  B,\beta\right)  }^{\left(  C,\gamma\right)
}T_{\left(  C,\gamma\right)  }.\label{s2'}%
\end{equation}
where $T_{\left(  A,\gamma\right)  }=\lambda_{\gamma}T_{A}$ is a basis of $G.$
The set (\ref{s1}) with the composition law (\ref{s2'}) is called a S-expanded
Lie algebra. \ This algebra is a Lie algebra structure defined over the vector
space obtained by taking ord $S$ copies of $G$
\[
\mathfrak{G:}\oplus_{\alpha\in S}W_{\alpha}\text{ }\left(  \mathbf{W}_{\alpha
}\approx\mathcal{G}\text{, }\mathbf{\forall\alpha}\right)
\]
$\dim G=ordS\times\dim G$ by means of the structure constants
\begin{equation}
C_{\left(  A,\alpha\right)  \left(  B,\beta\right)  }^{\left(  C,\gamma
\right)  }=C_{AB}^{C}\delta_{\alpha\beta}^{\gamma}\label{s2''}%
\end{equation}
where $\delta$ is the Kronecker symbol and the subindex $\alpha,\beta\in S$
denotes the inner compostion in $S$ so that $\delta_{\alpha\beta}^{\gamma}=1$
when $\alpha\beta=\gamma$ in $S$ and zero otherwise. The constants $C_{\left(
A,\alpha\right)  \left(  B,\beta\right)  }^{\left(  C,\gamma\right)  }$
defined by (\ref{s2''}) inherit the symmetry properties of $C_{AB}^{C}$ of $G$
\ by virtue of the abelian character of the $S$-product, and satisfy the
Jacobi identity.

In a nutshell, the S-expansion method can be seen as the natural
generalization of the In\"{o}n\"{u}-Wigner contraction, where instead of to
multiply the generators by a numerical parameter, we multiply the generator by
the elements of a Abelian semigroup.

\begin{theorem}
The product $\left[  ,\right]  _{S}$ defined in (\ref{s2'}) is also a Lie
product because it is linear, antisymmetric and satisfies the Jacobi identity.
This product defines a new Lie algebra characterized by the pair $\left(
\mathfrak{G,}\left[  ,\right]  _{S}\right)  $, and is called a $S$-expanded
Lie algebra.\textit{ }
\end{theorem}

\begin{proof}
Since the $S$-product is abelian, the product $\left[  ,\right]  _{S}$ defined
by \ (\ref{s2'}) inherits the symmetry properties of the \ \ product $\left[
,\right]  $ of \ $\mathcal{G},$ and satisfies the Jacobi identity. In fact,%
\begin{align}
&  \left[  \left[  T_{\left(  A_{1},\alpha_{1}\right)  },T_{\left(
A_{2},\alpha_{2}\right)  }\right]  _{S},T_{\left(  A_{3},\alpha_{3}\right)
}\right]  _{S}+\left[  \left[  T_{\left(  A_{2},\alpha_{2}\right)
},T_{\left(  A_{3},\alpha_{3}\right)  }\right]  _{S},T_{\left(  A_{1}%
,\alpha_{1}\right)  }\right]  _{S}\nonumber\\
&  +\left[  \left[  T_{\left(  A_{3},\alpha_{3}\right)  },T_{\left(
A_{1},\alpha_{1}\right)  }\right]  _{S},T_{\left(  A_{2},\alpha_{2}\right)
}\right]  _{S}\nonumber\\
&  =\frac{1}{1!}\frac{1}{2!}%
{\displaystyle\sum\limits_{\sigma\in S_{3}}}
\left(  -1\right)  ^{\pi\left(  \sigma\right)  }\left[  \left[  T_{\left(
A_{\sigma\left(  1\right)  },\alpha_{\sigma\left(  1\right)  }\right)
},T_{\left(  A_{\sigma\left(  2\right)  },\alpha_{\sigma\left(  2\right)
}\right)  }\right]  _{S},T_{\left(  A_{\sigma\left(  3\right)  }%
,\alpha_{\sigma\left(  3\right)  }\right)  }\right]  _{S}\nonumber\\
&  =\frac{1}{1!}\frac{1}{2!}%
{\displaystyle\sum\limits_{\sigma\in S_{3}}}
\left(  -1\right)  ^{\pi\left(  \sigma\right)  }\lambda_{\alpha_{\sigma\left(
1\right)  }}\lambda_{\alpha_{\sigma\left(  2\right)  }}\lambda_{\alpha
_{\sigma\left(  3\right)  }}\left[  \left[  T_{A_{\sigma\left(  1\right)  }%
},T_{A_{\sigma\left(  2\right)  }}\right]  ,T_{A_{\sigma\left(  3\right)  }%
}\right] \nonumber\\
&  =\frac{1}{1!}\frac{1}{2!}%
{\displaystyle\sum\limits_{\sigma\in S_{3}}}
\left(  -1\right)  ^{\pi\left(  \sigma\right)  }K_{\alpha_{\sigma\left(
1\right)  }\alpha_{\sigma\left(  2\right)  }\alpha_{\sigma\left(  3\right)  }%
}^{\gamma}\lambda_{\gamma}\left[  \left[  T_{A_{\sigma\left(  1\right)  }%
},T_{A_{\sigma\left(  2\right)  }}\right]  ,T_{A_{\sigma\left(  3\right)  }%
}\right] \nonumber\\
&  =K_{\alpha_{1}\alpha_{2}\alpha_{3}}^{\gamma}\lambda_{\gamma}\left(
\frac{1}{1!}\frac{1}{2!}%
{\displaystyle\sum\limits_{\sigma\in S_{3}}}
\left(  -1\right)  ^{\pi\left(  \sigma\right)  }\left[  \left[  T_{A_{\sigma
\left(  1\right)  }},T_{A_{\sigma\left(  2\right)  }}\right]  ,T_{A_{\sigma
\left(  3\right)  }}\right]  \right)  =0 \label{ps1}%
\end{align}
where we have used the commutativity $\left(  K_{\alpha_{\sigma\left(
1\right)  }\alpha_{\sigma\left(  2\right)  }\alpha_{\sigma\left(  3\right)  }%
}^{\gamma}=K_{\alpha_{1}\alpha_{2}\alpha_{3}}^{\gamma}\right)  $ and
associativity of the semigroup inner law, and the fact that the product
$\left[  ,\right]  $ satisfies the Jacobi identity. \ 
\end{proof}

From (\ref{ps1}) we can see that the Jacobi identity of the $S$-expanded Lie
algebra $\left(  S\otimes\mathcal{G}\mathfrak{,}\left[  ,\right]  _{S}\right)
$
\begin{equation}
\left(
\begin{array}
[c]{c}%
\left[  \left[  T_{\left(  A_{1},\alpha_{1}\right)  },T_{\left(  A_{2}%
,\alpha_{2}\right)  }\right]  _{S},T_{\left(  A_{3},\alpha_{3}\right)
}\right]  _{S}+\left[  \left[  T_{\left(  A_{2},\alpha_{2}\right)
},T_{\left(  A_{3},\alpha_{3}\right)  }\right]  _{S},T_{\left(  A_{1}%
,\alpha_{1}\right)  }\right]  _{S}\\
+\left[  \left[  T_{\left(  A_{3},\alpha_{3}\right)  },T_{\left(  A_{1}%
,\alpha_{1}\right)  }\right]  _{S},T_{\left(  A_{2},\alpha_{2}\right)
}\right]  _{S}%
\end{array}
\right)  =0 \label{ps2}%
\end{equation}
can be obtained if we multiply the Jacobi identity of the Lie algebra $\left(
\mathcal{G}\mathfrak{,}\left[  ,\right]  \right)  $ by $\lambda_{\alpha_{1}%
}\lambda_{\alpha_{2}}\lambda_{\alpha_{3}}$ or by the 3-selector $K_{\alpha
_{1}\alpha_{2}\alpha_{3}}^{\gamma}$:
\begin{equation}
\text{JI}\left(  S\otimes\mathcal{G}\mathfrak{,}\left[  ,\right]  _{S}\right)
=K_{\alpha_{1}\alpha_{2}\alpha_{3}}^{\gamma}\left(  \text{JI}\left(
\mathcal{G}\mathfrak{,}\left[  ,\right]  \right)  \right)  \text{ .}
\label{ps3}%
\end{equation}
Similarly, if multiply the Jacobi condition of the Lie algebra $\left(
\mathcal{G}\mathfrak{,}\left[  ,\right]  \right)  $
\begin{equation}
\frac{1}{2}\varepsilon_{A_{1}A_{2}A_{3}}^{B_{1}B_{2}B_{3}}C_{B_{1}B_{2}}%
^{C}C_{CB_{3}}^{D}=0 \label{ps4}%
\end{equation}
by $K_{\alpha_{1}\alpha_{2}\alpha_{3}}^{\beta}=K_{\alpha_{1}\alpha_{2}%
}^{\gamma}K_{\gamma\alpha_{3}}^{\beta}$ , we obtain the Jacobi condition of
the $S$-expanded Lie algebra $\left(  S\otimes\mathcal{G}\mathfrak{,}\left[
,\right]  _{S}\right)  $. In fact,%
\begin{equation}
K_{\alpha_{1}\alpha_{2}\alpha_{3}}^{\beta}\left(  \frac{1}{2}\varepsilon
_{A_{1}A_{2}A_{3}}^{B_{1}B_{2}B_{3}}C_{B_{1}B_{2}}^{C}C_{CB_{3}}^{D}\right)
=\frac{1}{2}\varepsilon_{A_{1}A_{2}A_{3}}^{B_{1}B_{2}B_{3}}K_{\alpha_{1}%
\alpha_{2}}^{\gamma}C_{B_{1}B_{2}}^{C}K_{\gamma\alpha_{3}}^{\beta}C_{CB_{3}%
}^{D}=0 \label{ps5}%
\end{equation}%
\begin{equation}
\frac{1}{2}\varepsilon_{A_{1}A_{2}A_{3}}^{B_{1}B_{2}B_{3}}C_{\left(
B_{1},\alpha_{1}\right)  \left(  B_{2},\alpha_{2}\right)  }^{\left(
C,\gamma\right)  }C_{\left(  C,\gamma\right)  \left(  B_{3},\alpha_{3}\right)
}^{\left(  D,\beta\right)  }=0. \label{ps6}%
\end{equation}

\subsection{$S$-Expansion of Lie Multialgebras}

The $S$-expansion method is based on combining the structure constants of
$\ \left(  \mathcal{G},\left[  ,...,\right]  \right)  $ with the inner law of
a semigroup $S$ to define the Lie bracket of a new, $S$-expanded multialgebra.
Let $S=\left\{  \lambda_{\alpha}\right\}  $ be a finite Abelian semigroup
endowed with a commutative and associative composition law $S\times
S\rightarrow S,$ $\left(  \lambda_{\alpha},\lambda_{\beta}\right)
\mapsto\lambda_{\alpha}\lambda_{\beta}=K_{\alpha\beta}^{\text{ \ \ \ \ }%
\gamma}\lambda_{\gamma}.$ The direct product $G=S\otimes G$ is defined as the
cartesian product set
\begin{equation}
\mathfrak{G}=S\times\mathcal{G}=\left\{  T_{\left(  A,\alpha\right)  }%
=\lambda_{\alpha}T_{A}\text{ : }\lambda_{\alpha}\in S\text{ , }T_{A}%
\in\mathcal{G}\right\}  \label{f0}%
\end{equation}
with the composition law $\left[  ,...,\right]  _{S}:G\overset{n}%
{\times...\times}G\rightarrow G$, defined by%
\[
\left[  T_{\left(  A_{1},\alpha_{1}\right)  },...,T_{\left(  A_{n},\alpha
_{n}\right)  }\right]  _{S}=\lambda_{\alpha_{1}}...\lambda_{\alpha_{n}}\left[
T_{A_{1}},...,T_{A_{n}}\right]
\]%
\begin{equation}
\left[  T_{\left(  A_{1},\alpha_{1}\right)  },...,T_{\left(  A_{n},\alpha
_{n}\right)  }\right]  _{S}=K_{\alpha_{1}...\alpha_{n}}^{\gamma}%
C_{A_{1}...A_{n}}^{C}\lambda_{\gamma}T_{C}=C_{\left(  A_{1},\alpha_{1}\right)
...\left(  A_{n},\alpha_{n}\right)  }^{\left(  C,\gamma\right)  }T_{\left(
C,\gamma\right)  } \label{f1}%
\end{equation}
where $T_{\left(  A_{i},\alpha_{i}\right)  }\in G$, $\forall i=1,...,n,$ and
\ \ $C_{\left(  A_{1},\alpha_{1}\right)  ...\left(  A_{n},\alpha_{n}\right)
}^{\left(  C,\gamma\right)  }=K_{\alpha_{1}...\alpha_{n}}^{\gamma}%
C_{A_{1}...A_{n}}^{C}.$ \ \ \ \ 

The set $G=S\times G$ (\ref{f0}) with the composition law (\ref{f1}) define a
new Lie multialgebra which will be called S-expanded Lie multialgebra. This
algebra is a Lie algebra structure defined over the vector space obtained by
taking $S$ copies of $G$ by means of the structure constant $C_{\left(
A_{1},\alpha_{1}\right)  ...\left(  A_{n},\alpha_{n}\right)  }^{\left(
C,\gamma\right)  }=K_{\alpha_{1}...\alpha_{n}}^{\gamma}C_{A_{1}...A_{n}}^{C}$
where $K_{\alpha_{1}...\alpha_{n}}^{\gamma}=K_{\alpha_{1}...\alpha_{n-1}%
}^{\sigma}K_{\sigma\alpha_{n}}^{\gamma}$. \ The structure constants
$C_{\left(  A_{1},\alpha_{1}\right)  ...\left(  A_{n},\alpha_{n}\right)
}^{\left(  C,\gamma\right)  }$ defined in (\ref{f1}) inherit the symmetry
properties of $C_{A_{1}...A_{n}}^{C}$ of $G$ by virtue of the abelian
character of the $S$-product.

\begin{theorem}
The product $\left[  ,...,\right]  _{S}$\ defined in (\ref{f1}))\ is
multilinear, antisymmetric and satisfies the generalized Jacobi identity
(GJI).%
\begin{equation}
a%
{\displaystyle\sum\limits_{\sigma\in S_{2n-1}}}
\left(  -1\right)  ^{\pi\left(  \sigma\right)  }\left[  \left[  T_{\left(
A_{\sigma\left(  1\right)  },\alpha_{\sigma\left(  1\right)  }\right)
},..,T_{\left(  A_{\sigma\left(  n\right)  },\alpha_{\sigma\left(  n\right)
}\right)  }\right]  _{S},T_{\left(  A_{\sigma\left(  n+1\right)  }%
,\alpha_{\sigma\left(  n+1\right)  }\right)  },..,T_{\left(  A_{\sigma\left(
2n-1\right)  },\alpha_{\sigma\left(  2n-1\right)  }\right)  }\right]  _{S}=0
\end{equation}
where
\[
a=\frac{1}{\left(  n-1\right)  !}\frac{1}{n!}%
\]

\end{theorem}

\begin{proof}
Since the $S$-product is abelian, the product $\left[  ,...,\right]  _{S}%
$\ defined by (\ref{f1}) inherits the symmetry properties of the \ \ product
$\left[  ,...,\right]  $ of $\ \left(  \mathcal{G},\left[  ,...,\right]
\right)  ,$ and satisfies the generalized Jacobi identity. In fact,
\begin{align}
&
{\displaystyle\sum\limits_{\sigma\in S_{2n-1}}}
\left(  -1\right)  ^{\pi\left(  \sigma\right)  }\left[  \left[  T_{\left(
A_{\sigma\left(  1\right)  },\alpha_{\sigma\left(  1\right)  }\right)
},...,T_{\left(  A_{\sigma\left(  n\right)  },\alpha_{\sigma\left(  n\right)
}\right)  }\right]  _{S},T_{\left(  A_{\sigma\left(  n+1\right)  }%
,\alpha_{\sigma\left(  n+1\right)  }\right)  },...,T_{\left(  A_{\sigma\left(
2n-1\right)  },\alpha_{\sigma\left(  2n-1\right)  }\right)  }\right]
_{S}\nonumber\\
&  =%
{\displaystyle\sum\limits_{\sigma\in S_{2n-1}}}
\left(  -1\right)  ^{\pi\left(  \sigma\right)  }\lambda_{\alpha_{\sigma\left(
1\right)  }}\ldots\lambda_{\alpha_{\sigma\left(  2n-1\right)  }}\left[
\left[  T_{A_{\sigma\left(  1\right)  }},...,T_{A_{\sigma\left(  n\right)  }%
}\right]  ,T_{A_{\sigma\left(  n+1\right)  }},...,T_{A_{\sigma\left(
2n-1\right)  }}\right] \nonumber\\
&  =%
{\displaystyle\sum\limits_{\sigma\in S_{2n-1}}}
\left(  -1\right)  ^{\pi\left(  \sigma\right)  }K_{\alpha_{\sigma\left(
1\right)  }\ldots\alpha_{\sigma\left(  2n-1\right)  }}^{\gamma}\lambda
_{\gamma}\left[  \left[  T_{A_{\sigma\left(  1\right)  }},...,T_{A_{\sigma
\left(  n\right)  }}\right]  ,T_{A_{\sigma\left(  n+1\right)  }}%
,...,T_{A_{\sigma\left(  2n-1\right)  }}\right] \nonumber\\
&  =K_{\alpha_{1}\ldots\alpha_{2n-1}}^{\gamma}\lambda_{\gamma}\left(
{\displaystyle\sum\limits_{\sigma\in S_{2n-1}}}
\left(  -1\right)  ^{\pi\left(  \sigma\right)  }\left[  \left[  T_{A_{\sigma
\left(  1\right)  }},...,T_{A_{\sigma\left(  n\right)  }}\right]
,T_{A_{\sigma\left(  n+1\right)  }},...,T_{A_{\sigma\left(  2n-1\right)  }%
}\right]  \right)  =0, \label{f3'}%
\end{align}
where we have used the commutativity $K_{\alpha_{\sigma\left(  1\right)
}\ldots\alpha_{\sigma\left(  2n-1\right)  }}^{\gamma}=K_{\alpha_{1}%
\ldots\alpha_{2n-1}}^{\gamma}$ and associativity of the semigroup inner law,
and the fact that the product $\left[  ,...,\right]  $ satisfies the
generalized Jacobi identity.
\end{proof}

From (\ref{f3'}) we can see that the Jacobi identity of the $S$-expanded Lie
multialgebra $\left(  S\otimes\mathcal{G}\mathfrak{,}\left[  ,...,\right]
_{S}\right)  $ can be obtained if we multiply the generalized Jacobi identity
of the Lie multialgebra $\left(  \mathcal{G}\mathfrak{,}\left[  ,...,\right]
\right)  $ by $K_{\alpha_{1}\ldots\alpha_{2n-1}}^{\gamma}.$

Similarly, if we multiply the generalized Jacobi condition of the Lie algebra
$\left(  \mathcal{G},\left[  ,...,\right]  \right)  $%
\begin{equation}
\varepsilon_{A_{1}...A_{2n-1}}^{B_{1}...B_{2n-1}}C_{B_{1}...B_{n}}%
^{C}C_{CB_{n+1}...B_{2n-1}}^{D}=0 \label{f4}%
\end{equation}
by $K_{\alpha_{1}\ldots\alpha_{2n-1}}^{\beta}=K_{\alpha_{1}\ldots\alpha_{n}%
}^{\gamma}K_{\gamma\alpha_{n+1}\ldots\alpha_{2n-1}}^{\beta}$, we obtain the
generalized Jacobi condition of the $S$-expanded Lie multialgebra $\left(
\mathfrak{G,}\left[  ,...,\right]  _{S}\right)  $. In fact,%
\begin{align}
K_{\alpha_{1}\ldots\alpha_{2n-1}}^{\beta}\left(  \varepsilon_{A_{1}%
...A_{2n-1}}^{B_{1}...B_{2n-1}}C_{B_{1}...B_{n}}^{C}C_{CB_{n+1}...B_{2n-1}%
}^{D}\right)   &  =0\label{f4_2}\\
\varepsilon_{A_{1}...A_{2n-1}}^{B_{1}...B_{2n-1}}K_{\alpha_{1}\ldots\alpha
_{n}}^{\gamma}C_{B_{1}...B_{n}}^{C}K_{\gamma\alpha_{n+1}\ldots\alpha_{2n-1}%
}^{\beta}C_{CB_{n+1}...B_{2n-1}}^{D}  &  =0\nonumber\\
\varepsilon_{A_{1}...A_{2n-1}}^{B_{1}...B_{2n-1}}C_{\left(  B_{1},\alpha
_{1}\right)  ...\left(  B_{n},\alpha_{n}\right)  }^{\left(  C,\gamma\right)
}C_{\left(  C,\gamma\right)  \left(  B_{n+1},\alpha_{n+1}\right)  ...\left(
B_{2n-1},\alpha_{n+1}\right)  }^{D}  &  =0. \label{f5}%
\end{align}

\subsection{Multialgebra $0_{S}$-Reduced}

When the semigroup has a zero element $0_{S}\in S$, it plays a somewhat
peculiar role in the $S$-expanded Lie multialgebra. Let us span $S$ in nonzero
elements $\lambda_{i},i=0,\cdot\cdot\cdot,N$, and a zero element
$\lambda_{N+1}=0_{S},$ i.e.,%

\begin{equation}
S=\underset{\ \ \ \ \ \ \ \ \ \ \ \ \ \ \ \ \ \ \ \ \ \ }{\left\{
\underset{\ \lambda_{i}\ }{\underbrace{\lambda_{0},\lambda_{1},\ldots
,\lambda_{N}}},\underset{0_{S}}{\underbrace{\lambda_{N+1}}}\right\}  }\text{.}
\label{cr1}%
\end{equation}
Then, the $2$-selector satisfies
\[
K_{N+1,i_{2},...,i_{n}}^{\ \ \ \ \ \ \ \ \ \ \ \ \ \ j}=K_{\underset
{r}{\underbrace{N+1,...,N+1}},i_{r+1},...,i_{n}}%
^{\ \ \ \ \ \ \ \ \ \ \ \ \ \ \ \ \ \ \ \ \ \ \ \ \ \ \ \ \ \ j}%
=K_{\underset{r}{\underbrace{N+1,...,N+1}},i_{r+1},...,i_{n}}%
^{\ \ \ \ \ \ \ \ \ \ \ \ \ \ \ \ \ \ \ \ \ \ \ \ \ \ \ \ \ N+1}=\cdot
\cdot\cdot=K_{N+1,...,N+1}^{\ \ \ \ \ \ \ \ \ \ \ \ \ \ j}=0
\]%
\begin{equation}
K_{N+1,i_{2},...,i_{n}}^{\ \ \ \ \ \ \ \ \ \ \ \ \ \ \ \ \ \ \ N+1}%
=K_{N+1,...,N+1}^{\ \ \ \ \ \ \ \ \ \ \ \ \ \ \ \ \ \ \ \ N+1}=1. \label{cr2}%
\end{equation}
Therefore, the $S$-expanded multialgebra $\left(  \mathfrak{G,}\left[
,...,\right]  _{S}\right)  $ can be split as%
\[
\left[  T_{\left(  A_{1},i_{1}\right)  },\ldots,T_{\left(  A_{n},i_{n}\right)
}\right]  _{S}=K_{i_{1},...,i_{n}}^{\ \ \ \ \ \ \ \ k}C_{A_{1}...A_{n}%
}^{\ \ \ \ \ \ \ \ C}T_{\left(  C,k\right)  }+K_{i_{1},...,i_{n}%
}^{\ \ \ \ \ \ \ \ N+1}C_{A_{1},...,A_{n}}^{\ \ \ \ \ \ \ \ \ \ \ C}T_{\left(
C,N+1\right)  }%
\]%
\[
\left[  T_{\left(  A_{1},N+1\right)  },T_{\left(  A_{2},i_{2}\right)  }%
,\ldots,T_{\left(  A_{n},i_{n}\right)  }\right]  _{S}=C_{A_{1},...,A_{n}%
}^{\ \ \ \ \ \ \ \ \ \ \ C}T_{\left(  C,N+1\right)  }%
\]%
\[
\vdots
\]%
\[
\left[  T_{\left(  A_{1},N+1\right)  },\ldots,T_{\left(  A_{r},N+1\right)
},T_{\left(  A_{r+1},i_{r+1}\right)  },\ldots,T_{\left(  A_{n},i_{n}\right)
}\right]  _{S}=C_{A_{1},...,A_{n}}^{\ \ \ \ \ \ \ \ \ \ \ C}T_{\left(
C,N+1\right)  }%
\]%
\[
\vdots
\]%
\begin{equation}
\left[  T_{\left(  A_{1},N+1\right)  },\ldots,T_{\left(  A_{n},N+1\right)
}\right]  _{S}=C_{A_{1},...,A_{n}}^{\ \ \ \ \ \ \ \ \ \ \ C}T_{\left(
C,N+1\right)  }. \label{cr3}%
\end{equation}

From (\ref{cr3}) we can see that $\left(  \mathfrak{G,}\left[  ,...,\right]
_{S}\right)  $ can be written as $\mathfrak{G}=V_{0}\oplus V_{1}$, with
$V_{0}=\left\{  T_{\left(  A,i\right)  }\right\}  $,$\ V_{1}=\left\{
T_{\left(  A,N+1\right)  }\right\}  $. From (\ref{cr3}) we also see that
\begin{equation}
\left[  V_{1},V_{0},...,V_{0}\right]  _{S}\subset V_{1} \label{cr4}%
\end{equation}%
\begin{equation}
\underset{r\text{-times\ \ \ \ \ \ \ \ \ \ \ \ \ \ \ \ \ \ \ \ \ \ \ \ \ }%
}{\left[  \underbrace{V_{1},...,V_{1}},V_{0},...,V_{0}\right]  _{S}}\subset
V_{1}\text{, \ \ con }r=1,...,n. \label{cr5}%
\end{equation}
This means that the commutation relations
\[
\left[  T_{\left(  A_{1},i_{1}\right)  },\ldots,T_{\left(  A_{n},i_{n}\right)
}\right]  _{S}=K_{i_{1},...,i_{n}}^{\ \ \ \ \ \ \ \ k}C_{A_{1}...A_{n}%
}^{\ \ \ \ \ \ \ \ C}T_{\left(  C,k\right)  }%
\]
are those of a reduced Lie multialgebra $\left(  \mathfrak{G,}\left[
,...,\right]  _{S}\right)  $. \ From (\ref{cr3}) we see that the reduction
procedure in this particular case is equivalent to imposing the condition
\[
T_{\left(  C,N+1\right)  }=0_{S}T_{C}=0\text{.}%
\]

The above considerations motivate the following definition:

\begin{definition}
Let $S$ be an Abelian semigroup with a zero element $0_{S}\in S$, and let
$\left(  S\otimes\mathcal{G}\mathfrak{,}\left[  ,...,\right]  \right)  $ be an
$S$-expanded multialgebra. The multialgebra obtained by imposing the condition
$0_{S}T_{A}=0$ on $\mathfrak{G}$ is called $\ $a $0_{S}$-reduced multialgebra
of $\mathfrak{G}$.
\end{definition}

\section{S-expansion of submultialgebras}

In this section is shown that there are at least two ways of extracting
smaller multialgebras from $\left(  S\otimes\mathcal{G}\mathfrak{,}\left[
,...,\right]  \right)  .$ The first one gives rise to a "resonant
submultialgebra" \ while the second produces reduced multialgebras of a
resonant submultialgebra.

\subsection{Resonant submultialgebras}

The general problem of finding submultialgebras from an $S$-expanded
multialgebra is a nontrivial one, which is met and solved in this section. In
order to provide a solution, one must have some information about the subspace
structure of $\mathcal{G}\mathfrak{,}\left[  ,...,\right]  .$ This information
is encoded in the following way:

Let $\mathcal{G}=\oplus_{p\in I}V_{p}$ be a decomposition of $\mathcal{G}$\ in
subspaces $V_{p}$, where $I$ is a set of indices. \ For each $\left(
p_{1},...,p_{n}\right)  \in I$ it is always possible to define $i_{\left(
p_{1},...,p_{n}\right)  }\subset I$ such that
\begin{equation}
\left[  V_{p_{1}},...,V_{p_{n}}\right]  \subset%
{\displaystyle\bigoplus\limits_{r\in i_{\left(  p_{1},...,p_{n}\right)  }}}
V_{r}\text{.} \label{msr1}%
\end{equation}
In this way, the subsets $\left\{  i_{\left(  p_{1},...,p_{n}\right)
}\right\}  $ store the information on the subspace structure of $\mathcal{G}$.

As for the Abelian semigroup $S$, this can always be decomposed as
$S=\cup_{p\in I}S_{p},$ where $S_{p}\subset S.$ In principle, this
decomposition is completely arbitrary; however, using the product from
definition $\left(  2.2\right)  $ of ref. \cite{sexpansion}, it is sometimes
possible to pick out a very particular choice of subset decomposition. This
choice is the subject of the following definition:

\begin{definition}
Let $\mathcal{G}=\oplus_{p\in I}V_{p}$ be a decomposition of $\mathcal{G}$\ in
subspaces $V_{p}$, with a structure described by the subsets $i_{\left(
p_{1},...,p_{n}\right)  },$ as in Eq.(\ref{msr1}). Let $S=\cup_{p\in I}S_{p}$
be a subset decomposition of the Abelian semigroup $S$ such that%
\begin{equation}
S_{p_{1}}\times S_{p_{2}}\times\cdot\cdot\cdot\cdot\times S_{p_{n}}\subset%
{\displaystyle\bigcap\limits_{r\in i_{\left(  p_{1},...,p_{n}\right)  }}}
S_{r}\text{.} \label{msr2}%
\end{equation}
When such a subset decomposition $S=\cup_{p\in I}S_{p}$ exists, then we say
that this decomposition is in resonance with the subspace decomposition of
$\mathcal{G}=\oplus_{p\in I}V_{p}$.
\end{definition}

\begin{theorem}
Let $\mathcal{G}=\oplus_{p\in I}V_{p}$ be a subspace decomposition of
$\mathcal{G},$with a structure described by Eq. \ (\ref{msr1}), and let
$S=\cup_{p\in I}S_{p}$ be a resonant subset decomposition of the Abelian
semigroup $S$, with the structure given in Eq.(\ref{msr2}). Define the
subspaces $W_{p}$ of $\mathfrak{G}=S\otimes\mathcal{G},$
\begin{equation}
W_{p}=S_{p}\otimes V_{p},\text{ }p\in I. \label{msr3}%
\end{equation}
Then,%
\begin{equation}
\mathfrak{G}_{R}=\oplus_{p\in I}W_{p} \label{msr4}%
\end{equation}
is called a resonant subalgebra of the S-expanded multialgebra $\mathfrak{G}%
=S\otimes\mathcal{G}$.
\end{theorem}

\begin{proof}
Using Eqs. (\ref{msr1}) and (\ref{msr2}) we have
\[
\left[  W_{p_{1}},...,W_{p_{n}}\right]  _{S}=\left[  S_{p_{1}}\otimes
V_{p_{1}},...,S_{p_{n}}\otimes V_{p_{n}}\right]  _{S}=\left(  S_{p_{1}}%
\times...\times S_{p_{n}}\right)  \otimes\left[  V_{p_{1}},...,V_{p_{n}%
}\right]
\]%
\begin{equation}
\subset\left(
{\displaystyle\bigcap\limits_{s\in i_{\left(  p_{1},...,p_{n}\right)  }}}
S_{s}\right)  \otimes\left(
{\displaystyle\bigoplus\limits_{r\in i_{\left(  p_{1},...,p_{n}\right)  }}}
V_{r}\right)  =%
{\displaystyle\bigoplus\limits_{r\in i_{\left(  p_{1},...,p_{n}\right)  }}}
\left(
{\displaystyle\bigcap\limits_{s\in i_{\left(  p_{1},...,p_{n}\right)  }}}
S_{s}\right)  \otimes V_{r}. \label{mrs6}%
\end{equation}
But, it is clear that for each $r\in i_{\left(  p_{1},...,p_{n}\right)  }$ one
can write
\begin{equation}%
{\displaystyle\bigcap\limits_{s\in i_{\left(  p_{1},...,p_{n}\right)  }}}
S_{s}\subset S_{r}. \label{msr7}%
\end{equation}
Then,
\[
\left[  W_{p_{1}},...,W_{p_{n}}\right]  _{S}\subset%
{\displaystyle\bigoplus\limits_{r\in i_{\left(  p_{1},...,p_{n}\right)  }}}
S_{r}\otimes V_{r}=%
{\displaystyle\bigoplus\limits_{r\in i_{\left(  p_{1},...,p_{n}\right)  }}}
W_{r}%
\]%
\[
\left[  W_{p_{1}},...,W_{p_{n}}\right]  _{S}\subset%
{\displaystyle\bigoplus\limits_{r\in i_{\left(  p_{1},...,p_{n}\right)  }}}
S_{r}\otimes V_{r}=%
{\displaystyle\bigoplus\limits_{r\in i_{\left(  p_{1},...,p_{n}\right)  }}}
W_{r}%
\]%
\begin{equation}
\left[  W_{p_{1}},...,W_{p_{n}}\right]  _{S}\subset%
{\displaystyle\bigoplus\limits_{r\in I}}
W_{r}=\mathfrak{G}_{R} \label{mrs8}%
\end{equation}

\end{proof}

Therefore, the algebra closes and $\mathfrak{G}_{R}$ is a submultialgebra of
$\mathfrak{G}$.

This theorem translates the difficult problem of finding subalgebras from an
$S$-expanded algebra $\mathfrak{G}=S\otimes\mathfrak{g}$ into that of finding
a resonant partition for the semigroup $S$.

Denoting the basis of $V_{p_{i}}$ by $\left\{  T_{a_{p_{i}}}\right\}  $,
$\lambda_{\alpha_{p_{i}}}\in S_{p_{i}}$ and $T_{\left(  a_{p_{i}}%
,\alpha_{p_{i}}\right)  }=\lambda_{\alpha_{p_{i}}}T_{a_{p_{i}}}\in W_{p_{i}}$
one can write%
\[
\left[  T_{\left(  a_{p_{1}},\alpha_{p_{1}}\right)  },...,T_{\left(  a_{p_{n}%
},\alpha_{p_{n}}\right)  }\right]  _{S}=C_{\left(  a_{p_{1}},\alpha_{p_{1}%
}\right)  ...\left(  a_{p_{n}},\alpha_{p_{n}}\right)  }%
^{\ \ \ \ \ \ \ \ \ \ \ \ \ \ \ \ \ \ \ \ \ \ \ \ \ \ \left(  c_{r},\gamma
_{r}\right)  }T_{\left(  c_{r},\gamma_{r}\right)  }\text{,}%
\]
which means that the structure constants of the resonant submultialgebra are
given by%
\[
C_{\left(  a_{p_{1}},\alpha_{p_{1}}\right)  ...\left(  a_{p_{n}},\alpha
_{p_{n}}\right)  }%
^{\ \ \ \ \ \ \ \ \ \ \ \ \ \ \ \ \ \ \ \ \ \ \ \ \ \ \left(  c_{r},\gamma
_{r}\right)  }=K_{\alpha_{p_{1}}...\alpha_{p_{n}}}^{\ \ \ \ \ \ \ \ \ \ \gamma
_{r}}C_{a_{p_{1}}...a_{p_{n}}}^{\ \ \ \ \ \ \ \ \ \ \ c_{r}}\text{.}%
\]

An interesting fact is that the S-expanded multialgebra "subspace structure"
encoded in $i_{\left(  p_{1},...,p_{n}\right)  }$ is the same as in the
original multialgebra, as can be observed fron Eq. (\ref{mrs8}).

\subsection{Reduced Multialgebras of a Resonant Submultialgebra}

The following theorem provides necessary conditions under which a reduced
multialgebra can be extracted from a resonant subalgebra:

\begin{theorem}
Let $\mathfrak{G}_{R}=\oplus_{p\in I}S_{p}\otimes V_{p}$ be a resonant
submultialgebra $\left(  \mathfrak{G,}\left[  ,...,\right]  _{S}\right)  $,
i.e., let Eqs. (\ref{msr1}) and (\ref{msr2}) be satisfied. Let $S_{p}=\hat
{S}_{p}\cup\check{S}_{p}$ be a partition of the subsets $S_{p}\subset S$ such
that%
\begin{equation}
\check{S}_{p_{i}}\cap\hat{S}_{p_{i}}=\phi\label{rrm1}%
\end{equation}%
\begin{equation}
\hat{S}_{p_{1}}\times\check{S}_{p_{2}}\times...\times\check{S}_{p_{n}}\subset%
{\displaystyle\bigcap\limits_{r\in i_{\left(  p_{1},...,p_{n}\right)  }}}
\hat{S}_{r}. \label{rrm2}%
\end{equation}
The conditions (\ref{rrm1}) and (\ref{rrm2}) induce the decomposition
$\mathfrak{G}_{R}=\mathfrak{\check{G}}_{R}\oplus\overset{\wedge}{\mathfrak{G}%
}_{R}$ on the resonant subalgebra, where
\begin{equation}
\mathfrak{\check{G}}_{R}=\oplus_{p\in I}\check{S}_{p}\otimes V_{p}
\label{rrm3}%
\end{equation}%
\begin{equation}
\overset{\wedge}{\mathfrak{G}}_{R}=\oplus_{p\in I}\hat{S}_{p}\otimes V_{p}.
\label{rrm4}%
\end{equation}
When conditions (\ref{rrm1}) and (\ref{rrm2}) hold, then%
\begin{equation}
\left[  \overset{\wedge}{\mathfrak{G}}_{R},\mathfrak{\check{G}}_{R}%
,...,\mathfrak{\check{G}}_{R}\right]  _{S}\subset\overset{\wedge}%
{\mathfrak{G}}_{R} \label{rrm5}%
\end{equation}
and therefore $\left\vert \mathfrak{\check{G}}_{R}\right\vert $\ corresponds
to a reduced algebra of $\ \mathfrak{G}_{R}$.
\end{theorem}

\begin{proof}
$\hat{W}_{p_{i}}=\hat{S}_{p_{i}}\otimes V_{p_{i}}$ and $\check{W}_{p_{i}%
}=\check{S}_{p_{i}}\otimes V_{p_{i}}$ . Then, using condition (\ref{rrm2}), we have:%

\begin{align*}
\left[  \hat{W}_{p_{1}},\check{W}_{p_{2}},...,\check{W}_{p_{n}}\right]  _{S}
&  =\left[  \hat{S}_{p_{1}}\otimes V_{p_{1}},\check{S}_{p_{2}}\otimes
V_{p_{2}},...,\check{S}_{p_{n}}\otimes V_{p_{n}}\right]  _{S}\\
&  =\left(  \hat{S}_{p_{1}}\times\check{S}_{p_{2}}\times...\times\check
{S}_{p_{n}}\right)  \otimes\left[  V_{p_{1}},V_{p_{2}},...,V_{p_{n}}\right] \\
&  \subset\left(
{\displaystyle\bigcap\limits_{s\in i_{\left(  p_{1},...,p_{n}\right)  }}}
\hat{S}_{s}\right)  \otimes\left(
{\displaystyle\bigoplus\limits_{r\in i_{\left(  p_{1},...,p_{n}\right)  }}}
V_{r}\right) \\
&  =%
{\displaystyle\bigoplus\limits_{r\in i_{\left(  p_{1},...,p_{n}\right)  }}}
\left(
{\displaystyle\bigcap\limits_{s\in i_{\left(  p_{1},...,p_{n}\right)  }}}
\hat{S}_{s}\right)  \otimes V_{r}.
\end{align*}
For each $r\in i_{\left(  p_{1},...,p_{n}\right)  }$ we have%
\[%
{\displaystyle\bigcap\limits_{s\in i_{\left(  p_{1},...,p_{n}\right)  }}}
\hat{S}_{s}\subset\hat{S}_{r}%
\]
so that,
\begin{align*}
\left[  \hat{W}_{p_{1}},\check{W}_{p_{2}},...,\check{W}_{p_{n}}\right]  _{S}
&  \subset%
{\displaystyle\bigoplus\limits_{r\in i_{\left(  p_{1},...,p_{n}\right)  }}}
\hat{S}_{r}\otimes V_{r}=%
{\displaystyle\bigoplus\limits_{r\in i_{\left(  p_{1},...,p_{n}\right)  }}}
\hat{W}_{r}\\
&  \subset%
{\displaystyle\bigoplus\limits_{r\in I}}
\hat{W}_{r}=\overset{\wedge}{\mathfrak{G}}_{R}.
\end{align*}
Thus $\left[  \hat{W}_{p_{1}},\check{W}_{p_{2}},...,\check{W}_{p_{n}}\right]
_{S}\subset\overset{\wedge}{\mathfrak{G}}_{R}$, i.e,%
\[
\left[  \overset{\wedge}{\mathfrak{G}}_{R},\mathfrak{\check{G}}_{R}%
,...,\mathfrak{\check{G}}_{R}\right]  _{S}\subset\overset{\wedge}%
{\mathfrak{G}}_{R}%
\]
and therefore $\left\vert \mathfrak{\check{G}}_{R}\right\vert $\ is a reduced
algebra of $\mathfrak{G}_{R}$.
\end{proof}

The structure constants for the reduced algebra $\left\vert \mathfrak{\check
{G}}_{R}\right\vert $\ are given by,%
\[
C_{\left(  a_{p_{1}},\alpha_{p_{1}}\right)  ...\left(  a_{p_{n}},\alpha
_{p_{n}}\right)  }%
^{\ \ \ \ \ \ \ \ \ \ \ \ \ \ \ \ \ \ \ \ \ \ \ \ \ \ \left(  c_{r},\gamma
_{r}\right)  }=K_{\alpha_{p_{1}}...\alpha_{p_{n}}}^{\ \ \ \ \ \ \ \ \ \ \gamma
_{r}}C_{a_{p_{1}}...a_{p_{n}}}^{\ \ \ \ \ \ \ \ \ \ \ c_{r}}%
\]
with $\alpha_{p_{i}}$, $\gamma_{r}$ such that $\lambda_{\alpha_{p_{i}}}%
\in\check{S}_{p_{i}}$ y $\lambda_{\gamma_{r}}\in\check{S}_{p_{r}}$.

\subsection{$S_{E}^{\left(  N\right)  }$-Expansion of Multialgebras}

\begin{definition}
Let us define $S_{E}^{\left(  N\right)  }$ as the semigroup of elements
\footnote{where the order of the multialgebra is denoted by $n$ and $N$
denotes the number of elements of the semigroup $S_{E}^{\left(  N\right)  }$.}%
\begin{equation}
S_{E}^{\left(  N\right)  }=\left\{  \lambda_{\alpha}\text{, }\alpha
=0,...,N+1\right\}  \label{b1}%
\end{equation}
provided with a multiplication rule%
\begin{equation}
\lambda_{\alpha}\lambda_{\beta}=\lambda_{H_{N+1}\left(  \alpha+\beta\right)
}=\delta_{H_{N+1}\left(  \alpha+\beta\right)  }^{\gamma}\lambda_{\gamma}
\label{b2}%
\end{equation}
where $H_{N+1}$ is defined as the function
\begin{equation}
H_{n}\left(  x\right)  =\left\{
\begin{array}
[c]{c}%
x\text{, when }x<n,\\
n\text{, when }x\geq n.
\end{array}
\right\}  . \label{b3}%
\end{equation}

\end{definition}

The two-selectors for $S_{E}^{\left(  N\right)  }$ read
\[
K_{\alpha\beta}^{\gamma}=\delta_{H_{N+1}\left(  \alpha+\beta\right)  }%
^{\gamma}%
\]

where $\delta_{\sigma}^{\rho}$ is the Kronecker delta.

The multiplication rule (\ref{b2}) can be directly generalized to
\begin{align}
\lambda_{\alpha_{1}}....\lambda_{\alpha_{n}}  &  =\lambda_{H_{N+1}\left(
\alpha_{1}+...+\alpha_{n}\right)  }=\delta_{H_{N+1}\left(  \alpha
_{1}+...+\alpha_{n}\right)  }^{\gamma}\lambda_{\gamma}\label{b4}\\
K_{\alpha_{1}...\alpha_{n}}^{\ \ \ \ \ \ \ \ \ \gamma}  &  =\delta
_{H_{N+1}\left(  \alpha_{1}+...+\alpha_{n}\right)  }^{\gamma}.\nonumber
\end{align}
From Eq.(\ref{b2}), we have that $\lambda_{N+1}$ is the zero element in
$\ S_{E}^{\left(  N\right)  },$ i.e., $\ \lambda_{N+1}=0_{S}$.

The corresponding $S$-expanded multialgebra is given by the following
commutation relation:%
\begin{equation}
\left[  T_{\left(  A_{1},\alpha_{1}\right)  },...,T_{\left(  A_{n},\alpha
_{n}\right)  }\right]  _{S}=\delta_{H_{N+1}\left(  \alpha_{1}+...+\alpha
_{n}\right)  }^{\gamma}C_{A_{1}...A_{n}}^{\ \ \ \ \ \ C}T_{\left(
C,\gamma\right)  }, \label{b5}%
\end{equation}
which implies that the structure constants for the $S_{E}^{\left(  N\right)
}$-expanded multialgebra can be written as
\begin{equation}
C_{\left(  A_{1},\alpha_{1}\right)  ...\left(  A_{n},\alpha_{n}\right)
}^{\ \ \ \ \ \ \ \ \ \ \ \ \ \ \ \ \ \ \ \left(  C,\gamma\right)  }%
=\delta_{H_{N+1}\left(  \alpha_{1}+...+\alpha_{n}\right)  }^{\gamma}%
C_{A_{1}...A_{n}}^{\ \ \ \ \ \ C} \label{b6}%
\end{equation}

with $\gamma,\alpha_{1},...,\alpha_{n}=0,\cdot\cdot\cdot,N+1.$ When the
condition of $0_{S}$-reduction is imposed, the Eq. (\ref{b6}) reduces to%
\[
C_{\left(  A_{1},i_{1}\right)  ...\left(  A_{n},i_{n}\right)  }%
^{\ \ \ \ \ \ \ \ \ \ \ \ \ \ \ \ \ \ \ \left(  C,k\right)  }=\delta
_{H_{N+1}\left(  i_{1}+...+i_{n}\right)  }^{k}C_{A_{1}...A_{n}}%
^{\ \ \ \ \ \ C}\text{.}%
\]

\section{Comments}

We have shown that the successful $S$-expansion of the Lie algebras method,
developed in ref. \cite{sexpansion}, can be generalized so as to obtain
expanded higher-order Lie algebras.

The main results of this paper are: the generalizations of the definitions of
Lie subalgebras and reduced Lie algebras to higher-order Lie subalgebras and
higher-order reduced Lie algebras; to generalize the S-expansion method and to
show that it is possible to obtain higher-order expanded Lie algebras, as well
as to probe that under determined conditions can be extracted
relevant higher-order Lie subalgebras from the S-expanded higher-order Lie algebras.

\bigskip

This work was supported in part by FONDECYT through Grants \#s 1080530 and
1070306 and in part by Direcci\'{o}n de Investigaci\'{o}n, Universidad de
Concepci\'{o}n through Grant \# 208.011.048-1.0. One of the authors (P.S) wish
to thank J.A. de Azcarraga for his kind hospitality at the Departament of
Theoretical Physics of Valencia University and many enlightening
discussions.Two of the authors (R.C. and N.M) were supported by grants from
the Comisi\'{o}n Nacional de Investigaci\'{o}n Cient\'{\i}fica y
Tecnol\'{o}gica CONICYT and from the Universidad de Concepci\'{o}n, Chile.

\section{Appendix A}

In this appendix we show that the realization (\ref{m2}) of the multibracket
satisfies the identity%

\begin{align}
&  \frac{1}{\left(  n-1\right)  !}\frac{1}{n!}%
{\displaystyle\sum\limits_{\sigma\in S_{2n-1}}}
\left(  -1\right)  ^{\pi\left(  \sigma\right)  }\left[  \left[  T_{A_{\sigma
\left(  1\right)  }},...,T_{A_{\sigma\left(  n\right)  }}\right]
,T_{A_{\sigma\left(  n+1\right)  }},...,T_{A_{\sigma\left(  2n-1\right)  }%
}\right] \label{ap1}\\
&  =\left\{
\begin{array}
[c]{c}%
0\text{ \ \ \ \ \ \ \ \ \ \ \ \ \ \ \ \ \ \ \ \ \ \ \ \ \ \ \ ,\ }n\text{
even}\\
n\left[  T_{A_{1}},...,T_{A_{2n-1}}\right]  \text{,\ \ \ \ \ \ }n\text{ odd.}%
\end{array}
\right\} \nonumber
\end{align}
which can be re-written in the following way:%
\begin{align}
&  \frac{1}{\left(  n-1\right)  !}\frac{1}{n!}\varepsilon_{A_{1}...A_{2n-1}%
}^{B_{1}...B_{2n-1}}\left[  \left[  T_{B_{1}},...,T_{B_{n}}\right]
,T_{B_{n+1}},...,T_{B_{2n-1}}\right] \label{ap2}\\
&  =\left\{
\begin{array}
[c]{c}%
0\text{ \ \ \ \ \ \ \ \ \ \ \ \ \ \ \ \ \ \ \ \ \ \ \ \ \ \ \ ,\ }n\text{
even}\\
nn!\left(  n-1\right)  !\left[  T_{A_{1}},...,T_{A_{2n-1}}\right]
\text{,\ }n\text{ odd.}%
\end{array}
\right\}  .\nonumber
\end{align}
In fact, if
\begin{equation}
\varphi=\varepsilon_{A_{1}...A_{2n-1}}^{B_{1}...B_{2n-1}}\left[  \left[
T_{B_{1}},...,T_{B_{n}}\right]  ,T_{B_{n+1}},...,T_{B_{2n-1}}\right]  \text{,}
\label{ap2_2}%
\end{equation}
then%
\begin{align}
\varphi &  =\varepsilon_{A_{1}...A_{2n-1}}^{B_{1}...B_{2n-1}}\left[
\varepsilon_{B_{1}...B_{n}}^{C_{1}...C_{n}}T_{C_{1}}...T_{C_{n}},T_{B_{n+1}%
},...,T_{B_{2n-1}}\right] \label{ap3}\\
&  =\varepsilon_{A_{1}...A_{2n-1}}^{B_{1}...B_{2n-1}}\varepsilon
_{B_{1}...B_{n}}^{C_{1}...C_{n}}\left[  T_{C_{1}}...T_{C_{n}},T_{B_{n+1}%
},...,T_{B_{2n-1}}\right] \nonumber\\
&  =n!\varepsilon_{A_{1}..................A_{2n-1}}^{C_{1}...C_{n}%
B_{n+1}...B_{2n-1}}\left[  T_{C_{1}}...T_{C_{n}},T_{B_{n+1}},...,T_{B_{2n-1}%
}\right] \nonumber
\end{align}
where we have used Eq.(\ref{m2}) and the property%
\begin{equation}
\varepsilon_{h_{1}...h_{r}}^{i_{1}...i_{r}}B^{h_{1}...h_{r}}=r!B^{i_{1}%
...i_{r}}\text{.} \label{ap4}%
\end{equation}
We consider now the multibracket $\left[  T_{C_{1}}...T_{C_{n}},T_{B_{n+1}%
},...,T_{B_{2n-1}}\right]  $. The expression $T_{C_{1}}...T_{C_{n}}$ is the
matrix product of $n$ elements, and therefore is a mapping onto another
element of $\mathcal{G},$ which must be antisymmetrized together with
$T_{B_{n+1}},...,T_{B_{2n-1}}$. Thus, we can write%
\begin{align}
&  \left[  T_{C_{1}}...T_{C_{n}},T_{B_{n+1}},...,T_{B_{2n-1}}\right]
\label{ap5}\\
&  =\varepsilon_{B_{n+1}...B_{2n-1}}^{C_{n+1}...C_{2n-1}}%
{\displaystyle\sum\limits_{s=0}^{n-1}}
\left(  -1\right)  ^{s}T_{C_{n+1}}...T_{C_{n+s}}T_{C_{1}}...T_{C_{n}%
}T_{C_{n+s+1}}...T_{C_{2n-1}}\nonumber
\end{align}
where the $n-1$ elements $T_{B_{n+1}},...,T_{B_{2n-1}}$ are antisymmetrized
with the contraction with $\varepsilon_{B_{n+1}...B_{2n-1}}^{C_{n+1}%
...C_{2n-1}}$ and the element $T_{C_{1}}...T_{C_{n}}$ is is antisymmetrized
with $%
{\textstyle\sum}
$ . Introducing these results into (\ref{ap3}) we have%
\begin{align}
\varphi &  =n!\varepsilon_{A_{1}..................A_{2n-1}}^{C_{1}%
...C_{n}B_{n+1}...B_{2n-1}}\varepsilon_{B_{n+1}...B_{2n-1}}^{C_{n+1}%
...C_{2n-1}}\label{ap6}\\
&  \times%
{\displaystyle\sum\limits_{s=0}^{n-1}}
\left(  -1\right)  ^{s}T_{C_{n+1}}...T_{C_{n+s}}T_{C_{1}}...T_{C_{n}%
}T_{C_{n+s+1}}...T_{C_{2n-1}}\nonumber\\
&  =n!\left(  n-1\right)  !\varepsilon_{A_{1}...A_{2n-1}}^{C_{1}...C_{2n-1}%
}\nonumber\\
&  \times%
{\displaystyle\sum\limits_{s=0}^{n-1}}
\left(  -1\right)  ^{s}T_{C_{n+1}}...T_{C_{n+s}}T_{C_{1}}...T_{C_{n}%
}T_{C_{n+s+1}}...T_{C_{2n-1}}\nonumber\\
&  =n!\left(  n-1\right)  !\nonumber\\
&  \times%
{\displaystyle\sum\limits_{s=0}^{n-1}}
\left(  -1\right)  ^{s}\varepsilon_{A_{1}...A_{2n-1}}^{C_{1}...C_{2n-1}%
}T_{C_{n+1}}...T_{C_{n+s}}T_{C_{1}}...T_{C_{n}}T_{C_{n+s+1}}...T_{C_{2n-1}%
}\nonumber
\end{align}
where we have used the identity (\ref{ap4}). Since%
\begin{align}
&  \varepsilon_{A_{1}...A_{2n-1}}^{C_{1}...C_{2n-1}}T_{C_{n+1}}...T_{C_{n+s}%
}T_{C_{1}}...T_{C_{n}}T_{C_{n+s+1}}...T_{C_{2n-1}}\label{ap7}\\
&  =\left(  -1\right)  ^{s}\varepsilon_{A_{1}...A_{2n-1}}^{C_{1}...C_{2n-1}%
}T_{C_{1}}T_{C_{n+1}}...T_{C_{n+s}}T_{C_{2}}...T_{C_{n}}T_{C_{n+s+1}%
}...T_{C_{2n-1}}\nonumber\\
&  =\left(  -1\right)  ^{s}\left(  -1\right)  ^{s}\varepsilon_{A_{1}%
...A_{2n-1}}^{C_{1}...C_{2n-1}}T_{C_{1}}T_{C_{2}}T_{C_{n+1}}...T_{C_{n+s}%
}T_{C_{3}}...T_{C_{n}}T_{C_{n+s+1}}...T_{C_{2n-1}}\nonumber\\
&  \vdots\nonumber\\
&  =\left(  -1\right)  ^{ns}\varepsilon_{A_{1}...A_{2n-1}}^{C_{1}...C_{2n-1}%
}T_{C_{1}}...T_{C_{n}}T_{C_{n+1}}...T_{C_{n+s}}T_{C_{n+s+1}}...T_{C_{2n-1}%
}\nonumber\\
&  =\left(  -1\right)  ^{ns}\varepsilon_{A_{1}...A_{2n-1}}^{C_{1}...C_{2n-1}%
}T_{C_{1}}...T_{C_{2n-1}},\nonumber
\end{align}
we have that (\ref{ap6}) takes the form%
\begin{align*}
\varphi &  =n!\left(  n-1\right)  !%
{\displaystyle\sum\limits_{s=0}^{n-1}}
\left(  -1\right)  ^{s}\left(  -1\right)  ^{ns}\varepsilon_{A_{1}...A_{2n-1}%
}^{C_{1}...C_{2n-1}}T_{C_{1}}...T_{C_{2n-1}}\\
&  =n!\left(  n-1\right)  !\varepsilon_{A_{1}...A_{2n-1}}^{C_{1}...C_{2n-1}%
}T_{C_{1}}...T_{C_{2n-1}}%
{\displaystyle\sum\limits_{s=0}^{n-1}}
\left(  -1\right)  ^{s}\left(  -1\right)  ^{ns}\\
&  =n!\left(  n-1\right)  !\left[  T_{A_{1}},...,T_{A_{2n-1}}\right]
{\displaystyle\sum\limits_{s=0}^{n-1}}
\left(  -1\right)  ^{s\left(  n+1\right)  }.
\end{align*}
It is direct to check that%
\[%
{\displaystyle\sum\limits_{s=0}^{n-1}}
\left(  -1\right)  ^{s\left(  n+1\right)  }=\left\{
\begin{array}
[c]{c}%
0\text{, for }n\text{ even\ \ \ \ }\\
n\text{, for }n\text{ odd}%
\end{array}
\right\}  \text{.}%
\]
Using (\ref{ap2_2}) we find%
\begin{align*}
&  \frac{1}{n!}\frac{1}{\left(  n-1\right)  !}\varepsilon_{A_{1}...A_{2n-1}%
}^{B_{1}...B_{2n-1}}\left[  \left[  T_{B_{1}},...,T_{B_{n}}\right]
,T_{B_{n+1}},...,T_{B_{2n-1}}\right] \\
&  =\left\{
\begin{array}
[c]{c}%
0\text{, \ \ \ \ \ \ \ \ \ \ \ \ \ \ \ \ \ \ \ \ \ \ \ for }n\text{
even\ \ \ }\\
n\left[  T_{A_{1}},...,T_{A_{2n-1}}\right]  \text{, for }n\text{ odd}%
\end{array}
\right\}
\end{align*}
or
\begin{align*}
&  \frac{1}{\left(  n-1\right)  !}\frac{1}{n!}%
{\displaystyle\sum\limits_{\sigma\in S_{2n-1}}}
\left(  -1\right)  ^{\pi\left(  \sigma\right)  }\left[  \left[  T_{A_{\sigma
\left(  1\right)  }},...,T_{A_{\sigma\left(  n\right)  }}\right]
,T_{A_{\sigma\left(  n+1\right)  }},...,T_{A_{\sigma\left(  2n-1\right)  }%
}\right] \\
&  =\left\{
\begin{array}
[c]{c}%
0\text{, \ \ \ \ \ \ \ \ \ \ \ \ \ \ \ \ \ \ \ \ \ \ \ for }n\text{
even\ \ \ }\\
n\left[  T_{A_{1}},...,T_{A_{2n-1}}\right]  \text{, for }n\text{ odd}%
\end{array}
\right\}  \text{.}%
\end{align*}

\end{document}